\documentclass[english]{article}
\usepackage[T1]{fontenc}
\usepackage[latin9]{inputenc}
\usepackage{geometry}
\usepackage{amsmath}
\usepackage{amsthm}
\usepackage{amssymb}

\makeatletter
\theoremstyle{plain}
\newtheorem{thm}{\protect\theoremname}
\theoremstyle{definition}
\newtheorem{defn}[thm]{\protect\definitionname}
\theoremstyle{remark}
\newtheorem{claim}[thm]{\protect\claimname}
\theoremstyle{plain}
\newtheorem{cor}[thm]{\protect\corollaryname}

\usepackage{fullpage}

\makeatother

\usepackage{babel}
\providecommand{\claimname}{Claim}
\providecommand{\corollaryname}{Corollary}
\providecommand{\definitionname}{Definition}
\providecommand{\theoremname}{Theorem}

\begin{document}
\title{A note on differentially private clustering with large additive error}
\author{Huy L. Nguyen\\
Northeastern University}
\maketitle
\begin{abstract}
In this note, we describe a simple approach to obtain a differentially
private algorithm for $k$-clustering with nearly the same multiplicative
factor as any non-private counterpart at the cost of a large polynomial
additive error. The approach is the combination of a simple geometric
observation independent of privacy consideration and any existing
private algorithm with a constant approximation.
\end{abstract}

\section{Clustering in low dimensions}

In this note, we consider the problem of finding an approximate clustering
solution with differential privacy in Euclidean space. The problem
has been studied extensively with many different objective functions.
Some of the popular ones include the $k$-median objective and the
$k$-mean objective. Recently the work \cite{GKM20} gave algorithms
for these objectives achieving almost the same multiplicative error
as any non-private counterpart and a large polynomial additive error.
In this note, we describe a simple alternative approach to achieve
a similar result. For concreteness, we focus on the $k$-median objective
but a similar proof also works for $k$-mean objective.
\begin{defn}
In the Euclidean $k$-median problem, we are given a dataset $D$
of $n$ points in $\mathbb{R}^{d}$. The goal is to find a set $S$
of $k$ centers to minimize the following objective:

\[
\min_{S}\sum_{p\in D}d(p,S)=\min_{S}\sum_{p\in D}\min_{c\in S}d(p,c)
\]

where $d(p,q)$ denotes the Euclidean distance between two points
$p$ and $q$. We use $d(p,S)$ as the shorthand for $\min_{q\in S}d(p,q)$.
\end{defn}

A major part of their work is in developing a private bi-criteria
algorithm for points in $\mathbb{R}^{d}$ with $poly\left(k,\log n,2^{d}\right)$
centers and clustering cost at most $\epsilon$ times the optimal
cost plus a polynomial additive error. We show that this result can
be obtained using a simple observation independent of privacy consideration.
Note that the observation holds more generally for metric spaces with
doubling dimension $d$.
\begin{claim}
Consider a dataset $D$ of $n$ points in $\mathbb{R}^{d}$ and a
constant $\epsilon\in(0,1/2]$. Let $O_{k}$ be the optimal $k$-median
solution and $OPT_{k}$ be the optimal $k$-median cost for the dataset.
Then for a certain $k'=k(1/\epsilon)^{O(d)}\log(n/\epsilon)$, we
have $OPT_{k'}\le O\left(\epsilon OPT_{k}\right)$.
\end{claim}

\begin{proof}
Suppose $O_{k}=\{c_{1},\ldots,c_{k}\}$ and suppose the optimal cost
is $Rn$. We will construct a new solution $S$ with $k'$ centers.
Let $T$ be the set of exponentially growing thresholds $T=\{\epsilon R,\epsilon R(1+\epsilon),\epsilon R(1+\epsilon)^{2},\ldots,nR\}$.
For each center $c_{i}$ and threshold $t\in T$, we cover the ball
$B(c_{i},t)$ (the ball centered at $c_{i}$with radius $t$) using
balls of radius $\epsilon t$ and include all the centers in the solution
$S$. We also include all $c_{i}$ in $S$. It is clear that $|S|=k(1/\epsilon)^{O(d)}|T|=k(1/\epsilon)^{O(d)}\log(n/\epsilon)$.

Next we show that the clustering cost of $S$ is at most $O(\epsilon Rn)$.
Consider a point $p$ in the dataset at distance $r=d(p,O_{k})$ from
its nearest center $c_{i}$ in $O_{k}$. If $r\le\epsilon R$ then
we just note that its distance to the nearest center in $S$ is also
at most $r$ (since $c_{i}\in S$). If $\epsilon R<r\le nR$ then
consider the minimum threshold $t\in T$ such that $t\ge r$. Since
$p\in B(c_{i},t)$, we include a center at distance at most $\epsilon t$
from $p$. By the minimality of $t$, we have $t\le(1+\epsilon)r$.
Thus, $p$ is at most $(1+\epsilon)\epsilon r$ away from some center
in $S$. The total clustering cost for $S$ is bounded by

\[
\sum_{p\in D}d(p,S)\le\left(\sum_{p\in D,d(p,O_{k})>\epsilon R}(1+\epsilon)\epsilon\cdot d(p,O_{k})\right)+n\epsilon R\le(1+\epsilon)\epsilon nR+n\epsilon R\le3n\epsilon R
\]
\end{proof}
Combining the above observation with an arbitrary private constant
approximation algorithm for $k$ median such as \cite{KS18} we obtain
the following result:
\begin{cor}
There is a $(\epsilon_{p},\delta_{p})$-differentially private algorithm
that works on data in the unit ball in $\mathbb{R}^{d}$ and outputs
$k'=k(1/\epsilon)^{O(d)}\log(n/\epsilon)$ centers such that the $k'$-median
clustering cost is at most $O(\epsilon OPT_{k})+poly\left(k,\log n,(1/\epsilon)^{d}\right)\log(1/\delta_{p})/\epsilon_{p}$
with probability at least $1-1/n^{2}$.
\end{cor}

\section{Clustering in high dimensions}

For completeness, we include a brief description of the remaining
steps to obtain an approximate solution using the bi-criteria solution.
\begin{thm}
Suppose there is a non-private algorithm with $\alpha$ approximation
for $k$-median in $\mathbb{\mathbb{R}}^{d}$. As a consequence, there
is an $(\epsilon_{p},\delta_{p})$-private algorithm for data in $B(0,1)$
that finds a solution with $k$-median cost $(\alpha+O(\epsilon))OPT_{k}+poly\left((k/\epsilon)^{\log(1/\epsilon)/\epsilon^{2}},\log n\right)\cdot d\log(1/\delta_{p})/\epsilon_{p}$
with probability $1-1/k$.
\end{thm}

\begin{proof}
The algorithm follows similar steps as those of Balcan et al. for
$k$-means \cite{BDLMZ17}.
\end{proof}
\begin{enumerate}
\item Project the data to $d'=O(\epsilon^{-2}\log k)$ dimensions and project
the results to the ball $B(0,\log n)$.
\item Run a $(\epsilon_{p}/3,\delta_{p})$-private algorithm on the projected
data to find a bi-criteria solution with $k'=k(1/\epsilon)^{O(d')}\log(n/\epsilon)$
centers.
\item Use the Laplace mechanism to compute the approximate number of points
assigned to each center.
\item Run a non-private algorithm on a new dataset where the points are
the $k'$ centers and each center has multiplicity equal to the approximate
number of points assigned to it i.e. snapping each point to its nearest
center.
\item Partition the data according to each point's closest center produced
in step 4. For each cluster, use a private algorithm to recover an
approximate optimal center in the original high dimensions.
\end{enumerate}
By \cite{MMR19}, projecting to $d'=O(\epsilon^{-2}\log k)$ dimensions
using a random Gaussian matrix preserves the clustering cost within
a $1+\epsilon$ factor with probability $1-1/k^{2}$. By the standard
argument using the concentration of the $\chi^{2}$ distribution,
with probability $1-1/n^{2}$, the resulting points are also contained
within the ball $B(0,\log n)$ in $\mathbb{R}^{d'}$. Thus, with probability
$1-1/n^{2}$, the step of projecting to the ball $B(0,\log n)$ does
not move any point. The reason we include this step is to protect
privacy in the low probability event where the projection fails.

In step 2, the algorithm produces a solution with cost $O(\epsilon OPT_{k})+poly\left((k/\epsilon)^{\log(1/\epsilon)/\epsilon^{2}},\log n\right)\log(1/\delta_{p})/\epsilon_{p}$.

In step 3, the number of points at each center is accurate up to additive
error $poly\left((k/\epsilon)^{\log(1/\epsilon)/\epsilon^{2}},\log n\right)/\epsilon_{p}$
per count. Thus, the new dataset has optimal $k$-median cost $(1+O(\epsilon))OPT_{k}+poly\left((k/\epsilon)^{\log(1/\epsilon)/\epsilon^{2}},\log n\right)\log(1/\delta_{p})/\epsilon_{p}$
(the original optimal cost plus the increase due to snapping points
to centers and the inaccurate counts).

In step 4, the non-private clustering algorithm produces a solution
with cost $\alpha(1+O(\epsilon))OPT_{k}+\alpha poly\left((k/\epsilon)^{\log(1/\epsilon)/\epsilon^{2}},\log n\right)\log(1/\delta_{p})/\epsilon_{p}$.

In step 5, we can use the private convex empirical risk minimization
algorithm \cite{BST14} to compute the approximate 1-median solution
for each cluster separately. The algorithm works for convex Lipschitz
risk function and the 1-median cost function is a convex 1-Lipschitz
function. The algorithm has additive error $poly(k)\cdot d/\epsilon_{p}$.

The result follows by adding up the costs in steps 4 and 5.

\bibliographystyle{plain}
\bibliography{kmedian}

\end{document}